\documentclass[11pt]{article}

\usepackage{fullpage}
\usepackage[linktocpage]{hyperref}
\usepackage[dvipsnames]{xcolor}
\definecolor{ForestGreen}{rgb}{0.0333,0.4451,0.0333}
\definecolor{DarkRed}{rgb}{0.65,0,0}
\definecolor{Red}{rgb}{1,0,0}
\hypersetup{
    unicode=false,          
    colorlinks=true,        
    linkcolor=BrickRed,          
    citecolor=OliveGreen,        
    filecolor=magenta,      
    urlcolor=cyan           
}

\newcommand{\eps}{\epsilon}
\newcommand{\mcP}{\mathcal{P}}
\newcommand{\poly}{\operatorname{poly}} 

\newcommand{\nextme}{\operatorname{next}}

\usepackage{amsmath}
\usepackage{amssymb}
\usepackage{amsthm}
\usepackage{thmtools}
\usepackage{thm-restate}
\usepackage{bm}
\usepackage{graphicx}
\usepackage{caption}
\usepackage[margin = 1in]{geometry}
\usepackage{natbib}
\setcitestyle{numbers,square}
\usepackage{color}
\usepackage{media9}
\usepackage{subcaption}
\usepackage{comment}
\usepackage{enumitem}
\usepackage[]{hyperref}
\usepackage[nameinlink,capitalize,nosort]{cleveref}
\crefname{claim}{Claim}{Claims}
\usepackage{palatino}

\usepackage{algorithm}
\usepackage[noend]{algpseudocode}

\usepackage{tcolorbox}
\usepackage[colorinlistoftodos]{todonotes}

\definecolor{mygreen}{rgb}{0,0.65,0.66}

\newcommand{\OPT}{\operatorname{OPT}}

\graphicspath{{imgs/}}

\makeatletter
\def\mathcolor#1#{\@mathcolor{#1}}
\def\@mathcolor#1#2#3{%
  \protect\leavevmode
  \begingroup
    \color#1{#2}#3%
  \endgroup
}
\makeatother

\newcommand*{\vsepfbox}[1]{%
  \begingroup
    \sbox0{\fbox{#1}}%
    \setlength{\fboxrule}{0pt}%
    \mbox{\kern-\fboxsep\fbox{\unhbox0}\kern-\fboxsep}%
  \endgroup
}

\theoremstyle{plain} \numberwithin{equation}{section}
\newtheorem{theorem}{Theorem}[section]
\numberwithin{theorem}{section}

\newtheorem{lemma}[theorem]{Lemma}

\newtheorem{claim}[theorem]{Claim}

\newtheorem{definition}[theorem]{Definition}

\theoremstyle{definition}

\newtheorem{remark}[theorem]{Remark}

\usepackage{thmtools} 
\usepackage{thm-restate}

\DeclareMathOperator*{\argmin}{argmin}

\makeatletter
\newcommand{\subalign}[1]{%
  \vcenter{%
    \Let@ \restore@math@cr \default@tag
    \baselineskip\fontdimen10 \scriptfont\tw@
    \advance\baselineskip\fontdimen12 \scriptfont\tw@
    \lineskip\thr@@\fontdimen8 \scriptfont\thr@@
    \lineskiplimit\lineskip
    \ialign{\hfil$\m@th\scriptstyle##$&$\m@th\scriptstyle{}##$\hfil\crcr
      #1\crcr
    }%
  }%
}
\makeatother

\newcounter{note}[section]
\renewcommand{\thenote}{\thesection.\arabic{note}}
\newcommand{\enote}[1]{\refstepcounter{note}$\ll${\bf Ellis~\thenote:}{\sf \color{blue} #1}$\gg$\marginpar{\tiny\bf EH~\thenote}}

\title{Simple Length-Constrained Minimum Spanning Trees}
\author{ D Ellis Hershkowitz\\ Brown University\\ \and Richard Z Huang\\ Brown University }
\date{}
\begin{document}
\maketitle

\begin{abstract}
    In the length-constrained minimum spanning tree (MST) problem, we are given an $n$-node edge-weighted graph $G$ and a length constraint $h \geq 1$. Our goal is to find a spanning tree of $G$ whose diameter is at most $h$ with minimum weight. Prior work of Marathe et al.\ gave a poly-time algorithm which repeatedly computes maximum cardinality matchings of minimum weight to output a spanning tree whose weight is $O(\log n)$-approximate with diameter $O(\log n)\cdot h$.

    In this work, we show that a simple random sampling approach recovers the results of Marathe et al.---no computation of min-weight max-matchings needed! Furthermore, the simplicity of our approach allows us to tradeoff between the approximation factor and the loss in diameter: we show that for any $\epsilon \geq 1/\operatorname{poly}(n)$, one can output a spanning tree whose weight is $O(n^\epsilon / \epsilon)$-approximate with diameter $O(1/\epsilon)\cdot h$ with high probability in poly-time. This immediately gives the first poly-time $\operatorname{poly}(\log n)$-approximation for length-constrained MST whose loss in diameter is $o(\log n)$.
\end{abstract}
\setcounter{page}{1}

\section{Introduction}

The minimum spanning tree (MST) problem is among the most well-studied problems in algorithms. In \emph{MST} we are given an $n$-node graph $G = (V,E)$ and an edge weight function $w : E \to \mathbb{R}_{\geq 0}$. 
Our goal is to find a spanning tree $H$ of $G$ of minimum total edge weight $w(H) := \sum_{e \in H}w(e)$. Part of what has motivated the extensive study of MST is the fact that it models numerous problems of practical interest and, perhaps most notably, the fact that it gives a natural formalism for the problem of finding minimum-cost connected networks. 

However, connectivity is often not all that is desired of a network. If transmitting a message over an edge incurs some latency, then efficient communication requires that all nodes are not just connected but connected by short paths. Likewise, if transmissions over edges fail with some probability, then reliable communication requires that nodes are connected by short paths.

The length-constrained MST problem  aims to address these issues.\footnote{This problem has appeared under many names in the literature, including: hop-bounded MST \cite{haeupler2021tree, chekuri2024approximation}, bounded diameter MST \cite{ravi1998bicriteria,kapoorsarwat2007bdmst} and shallow-light trees \cite{kortsarz1999shallowlight}.} \emph{Length-constrained MST} is identical to MST but we are also given a length constraint $h \geq 1$ and  our output spanning tree $H$ must have diameter at most $h$; here, the diameter of $H$ is the maximum distance between two nodes in $H$ where we treat each edge as having length $1$. We let $\OPT_h$ be the weight of the optimal solution. Length-constrained MST and its variants have been extensively studied \cite{ravi1998bicriteria,bar2001generalized,kortsarz1999shallowlight,kapoorsarwat2007bdmst,chimanispoerhase2014shallowlight,segal2022bdmst,haeupler2021tree, filtser2022hop, chekuri2024approximation,hajiaghayi2009approximating,hajiaghayi2006approximating, kantor2009approximate,althaus2005approximating}.

Arguably the most notable work in this area is that of \cite{ravi1998bicriteria}, which showed that it is always possible to find a spanning tree of weight at most $O(\log n) \cdot \OPT_h$ whose diameter is at most $O(\log n) \cdot h$ in polynomial-time. Henceforth, we refer to an algorithm for length-constrained MST that returns a spanning tree of weight at most $\alpha \cdot \OPT_h$ and diameter at most $\beta \cdot h$ an $\alpha$-approximation with length slack $\beta$. Thus, \cite{ravi1998bicriteria} gave a polynomial-time $O(\log n)$-approximation with length slack $O(\log n)$.  

The algorithm of \cite{ravi1998bicriteria} is essentially $O(\log n)$ rounds of min-weight max-matching computations. Specifically, for vertices $u$ and $v$, let
\begin{align*}
    \mcP_h(s,t) := \{P = (s, \ldots, t) : |P| \leq h\}
\end{align*}
be all $h$-length paths---i.e.\ all paths with at most $h$ edges---between $s$ and $t$. Likewise, let 
\begin{align*}
    d^{(h)}(s,t) := \min \{w(P) : P \in \mcP_h(s,t)\}
\end{align*}
be the $h$-length-constrained distance between $s,t$ and $P_h(s,t)\in\mathcal{P}_h(s,t)$ be any $h$-length $s,t$ path with weight $d^{(h)}(s,t)$. For a subset of vertices $S \subseteq V$, let $G^{(h)}[S] := (S, E_S)$ be the complete graph on $S$ where the weight of edge $\{u,v\} \in E_S$ is $d^{(h)}(u,v)$. Initially, call all nodes of $V$ active and set $T=(V,\emptyset)$. Then, do the following until there is only a single active node: let $S$ be all active nodes; let $M$ be a maximum cardinality matching of $G^{(h)}[S]$ of minimum weight; for each $(u,v) \in M$ remove $v$ from $S$, thereby deactivating $v$, and add to $T$ the path $P_h(u,v)$. Lastly, return any shortest path tree of $T$ for an arbitrary root vertex. 

Since this algorithm reduces the number of active nodes by a constant each round, it only requires $O(\log n)$ rounds. The weight and diameter of the partial solution increases by at most $O(\OPT_h)$ and $O(h)$ respectively each round, giving the approximation and length slack of $O(\log n)$.

\subsection{Our Contributions}

In this work, we provide an algorithm that recovers the guarantees of \cite{ravi1998bicriteria} but is simpler and unlike that of \cite{ravi1998bicriteria}, allows for a tradeoff between approximation and length slack. We discuss these notions below.  Our algorithm is randomized and succeeds with high probability.\footnote{That is, with probability at least $1-1/\poly(n)$.}

\begin{enumerate}
    \item \textbf{Simpler:} We show that repeatedly sampling a uniformly random constant fraction of active nodes and then connecting unsampled active nodes to their $d^{(h)}$-nearest sampled node provides the same guarantees as \cite{ravi1998bicriteria}. In other words, an approach far less sophisticated than that of \cite{ravi1998bicriteria}---which does not require min-weight max-cardinality matching computations---can achieve the same results. 
    
    Specifically, each of the $O(\log n)$ rounds of our algorithm simply samples a random constant fraction of active nodes, connects all unsampled active nodes to their $d^{(h)}$-nearest sampled node, and then declares any unsampled node inactive. Like the algorithm of \cite{ravi1998bicriteria}, one can argue that each round reduces the number of active nodes by a constant factor and so after only $O(\log n)$ rounds this process ends with no more active nodes. Likewise, we show that the diameter and weight of our solution increases by an $h$ and $\OPT_h$ (in expectation), leading to an approximation and length slack of $O(\log n)$.
    
    \item \textbf{Tradeoff:} Our approach gives a natural way of trading off between approximation and length slack. Specifically, we can reduce the total number of rounds required until we have no more active nodes by decreasing our sampling probability: if we fix an $\eps \geq 1/\poly(n)$ and then sample each node with probability $1/n^{\eps}$ then we have no more active nodes after only about $1/\eps$ rounds. \Cref{alg:stars} gives a formal description of our algorithm.
    
    In the end, we give a poly-time $O(n^\eps / \eps)$-approximation with length slack $O(1/\eps)$ for any $\eps \geq 1/ \poly(n)$. Since modifying our sampling probability does not change how much we increase the diameter of our solution in each round, our algorithm has a length slack of only $O(1/\eps)$. On the other hand, this more aggressive sampling probability increases the weight of our solution: we show that the expected increase in weight of the procedure is $O(n^\eps \cdot \OPT_h)$ in each round. Showing this is the main challenge of our approach; more below.
    
     Note that this tradeoff allows us to improve the length slack of \cite{ravi1998bicriteria} to $o(\log n)$ while still obtaining a $\poly(\log n)$-approximation by, e.g.,\ letting $\eps = \Theta(\log \log n / \log n)$. To our knowledge, and as we later discuss in \Cref{sec:related}, this is the first poly-time  $\poly(\log n)$-approximation for length-constrained MST with length slack $o(\log n)$. Furthermore, this tradeoff is particularly noteworthy in the context of length-constrained MST as similar tradeoffs appear in many other length-constrained problems: for example, similar tradeoffs are known for spanners \cite{althofer1993sparse}, oblivious reconfigurable networks \cite{amir2022optimal, wilson2024breaking} the length and cut slack of length-constrained expander decompositions \cite{haeupler2022hop,haeupler2024new,haeupler2023parallel} and the arboricity of ``parallel greedy'' graphs \cite{haeupler2023parallel}.
\end{enumerate}
\noindent  The following summarizes our main result. 
\begin{restatable}{theorem}{thmmain}\label{thm:main}
    For any $\eps \geq 1/\poly(n)$, there is a poly-time $O(n^\epsilon/\epsilon)$-approximation for length-constrained MST with length slack $O(1/\eps)$ that succeeds with high probability.
\end{restatable}
\noindent The result of \cite{ravi1998bicriteria} can be recovered from the above by letting $\eps = \Theta(1 /\log n)$.  Our result generalizes in standard ways: it is easy to see that our algorithm works for graphs with general lengths, not just unit edge lengths. Further note that, as observed by \cite{ravi1998bicriteria}, it is easy to turn the above result into an $O(1/\eps)$-approximation for the version of length-constrained MST that treats the diameter as the objective and the weight as the constraint. Specifically, our result gives an $O(1/\eps)$-approximation for the problem of finding a minimum diameter spanning tree whose weight is no larger than some input weight $W$ with the slight caveat that we return a weight $W \cdot O(n^{\eps}/\eps)$ tree instead of a weight $W$ tree. 

Like in \cite{ravi1998bicriteria}, our algorithm and analysis can be easily adapted to work for the Steiner tree version of the problem. Our algorithm can further be implemented using $O(h)$ Bellman-Ford iterations for each vertex.

As mentioned, the main challenge in proving our result is showing that in each round of random sampling, we increase the weight of our solution by at most about $n^\eps \cdot \OPT_h$. The rough idea for doing so is as follows. We consider an Euler tour of an optimal solution\footnote{We say that an Euler tour of a tree is the Euler tour on the graph on the same vertices and edges, but each edge is duplicated into two arcs.}, unroll this tour into a cycle and then contract this cycle so that it only contains one copy of each active node; see \Cref{fig:eulertourcyclefull}. Then, suppose each unsampled active node ``walks'' clockwise along this cycle until it hits a sampled node. It is not too hard to see that the total distance that a node must walk in the cycle until it finds a sampled node upper bounds its $d^{(h)}$ distance from some sampled node in the original graph. It follows that the total distance walked by all active nodes until they hit a sampled node upper bounds the sum of the $d^{(h)}$ distances of each active node from some sampled node; in other words, it upper bounds what our algorithm pays in one round. The expected number of nodes that walk over a given edge of this cycle is at most $O(n^\eps)$ by our choice of sampling probability. Thus we pay at most $n^\eps$ times the total weight of the cycle which is, in turn, bounded by $O(\OPT_h)$ (since it was obtained by just unrolling and contracting an Euler tour of the optimal solution). We formalize this argument in \Cref{sec:analy}.

\subsection{Additional Related Work}\label{sec:related}

Our work adds to a considerable body of work on length-constrained MST in addition to that of  \cite{ravi1998bicriteria}. We summarize this work below and in \Cref{fig:related}.

\begin{figure}[H]\label{fig:related}
\begin{center}\renewcommand{\arraystretch}{1.25}
\begin{tabular}{|l|l|l|l|l|l|}
\hline
 Approximation & Length Slack &  Running Time &Notes & Citation  \\ \hline
 $O(\log n)$ & $O(\log n)$    & $\poly(n)$     & Computes Matchings &  \cite{ravi1998bicriteria} \\ \hline
$O(n^\eps \cdot \exp(1/\eps))$ &    $1$      &  $n^{1/\eps} \cdot \poly(n)$    & Recursive  &\cite{kortsarz1999shallowlight} \\ \hline
$O(n^\eps /\eps^2)$ &    $1$      &  $n^{1/\eps} \cdot \poly(n)$    & Recursive  &\cite{charikar1999approximation} \\ \hline
 $O(1/\eps)$ &  $O(n^\eps / \eps)$        &   $\poly(n)$   & Computes MSTs &\cite{kapoorsarwat2007bdmst} \\ \hline
 $O(\log^3 n)$ & $O(\log^3 n)$       &   $\poly(n)$   & LP-Competitive & \cite{chekuri2024approximation} \\ \hline
$O(n^\eps/\eps)$ & $O(1/\eps)$       & $\poly(n)$  &  Randomized  &  \textbf{This Work} \\ \hline
\end{tabular}
\end{center}\caption{A summary of work on length-constrained MST.}\label{fig:related}
\end{figure}

Two lines of work give algorithms for length-constrained MST that allow for tradeoffs. First, \cite{kortsarz1999shallowlight} showed how to give an approximation of $O(n^{\eps} \cdot \exp(1/\eps))$ with no length slack but with running time $n^{1/\eps} \cdot \poly(n)$ using an intricate recursive procedure for $\eps > 0$. 
A later paper \cite{charikar1999approximation} noted and fixed a bug in the former and improved the approximation guarantee to $O(n^{\eps}/\eps^2)$. Second, Kapoor and Sarwat \cite{kapoorsarwat2007bdmst} showed how, given $\eps > 0$, one can compute an $O(1/\eps)$-approximation with length slack $O(n^{\epsilon}/\epsilon)$.\footnote{Kapoor and Sarwat \cite{kapoorsarwat2007bdmst} parameterize their algorithm slightly differently; we describe their result here in a way that is amenable to comparison with our result.} 
The algorithm of Kapoor and Sarwat repeatedly merges bounded depth subtrees of the MST of the aforementioned graph $G^{(h)}[V]$ (and as such requires repeatedly computing an MST).

Note that it is easy to see that there are many regimes of $\epsilon$ for which our results are not implied by any of the previous work. For instance, our algorithm with $\eps = \Theta(\log \log n / \log n)$ gives the first poly-time $\poly (\log n)$-approximation with length slack $o(\log n)$. Observe that the guarantees of Kapoor and Sarwat \cite{kapoorsarwat2007bdmst} cannot provide length slack $o(\log n)$ since $n^{\eps}/\eps = \Omega(\log n)$ for any $\eps > 0$ (even if $\eps$ is a function of $n$).\footnote{We give a proof sketch of why $n^{\eps}/\eps = \Omega(\log n)$: if $\eps \leq 1/\log n$ then the result is trivial. On the other hand, if $\eps  \geq 1 / \log n$ then we can express $\eps$ as $f(n) / \log n$ for some function $f$ where $f(n) \geq 1$. It follows that $n^{\eps}/\eps \geq \log n \cdot  \frac{\exp(f(n))}{f(n)} \geq \log n$ where we used the fact that $f(n) \geq 1$.} Likewise, the algorithms of \cite{kortsarz1999shallowlight,charikar1999approximation} only give $\poly(\log n)$-approximations when $\eps \leq O(\log \log n / \log n)$ in which case the running time of their algorithm is $\omega(\poly(n))$.

A very recent work of \cite{chekuri2024approximation} showed that poly-time $\poly(\log n)$-approximations for length-constrained MST that are competitive with respect to the natural linear program are possible with $\poly(\log n)$ length slack by using using tree embeddings for length-constrained distances and length-constrained oblivious routing. On the hardness side, \cite{naor1997retractedshallowlight} showed that any poly-time $o(\log n)$-approximation requires length slack $\Omega(1)$ via a reduction from set cover.\footnote{We also note that this work claimed an $O(\log n)$-approximation with length slack $2$ (for even the directed case) but later retracted this due to an error (see e.g.\ \cite{chimanispoerhase2014shallowlight,haeupler2021tree}).}

We also describe notable work on some special cases and generalizations of length-constrained MST. \cite{segal2022bdmst} gave approximation algorithms when the ratio of the maximum and minimum edge weight is bounded; specifically they gave a poly-time $O(1+\frac{c_{\max}}{c_{\min}} \cdot \frac{h}{D(n-1)})$ approximation where $D$ is the (hop) diameter of the input graph and $c_{\max}$ and $c_{\min}$ are the max and min edge weights respectively---however, this approximation can be arbitrarily large if $c_{\max}\gg c_{\min}$. \cite{althaus2005approximating} showed that if the input graph is a complete graph that induces a metric then a poly-time $O(\log n)$-approximation is possible. Lastly, for the directed case, \cite{chimanispoerhase2014shallowlight} gave a poly-time $O(n^{1/2+\epsilon})$ approximation for any fixed constant $\eps > 0$ with $O(1)$ length slack by using a combination of randomized LP rounding and techniques similar to \cite{kortsarz1999shallowlight}. 

\section{Algorithm Description} 
We now formally describe our algorithm. We assume we are given a graph $G=(V,E)$ with an edge-weight function $w:E\to\mathbb{R}_{\geq0}$, diameter bound $h$, and parameter $\epsilon\geq1/\poly(n)$. In the algorithm, we initialize our partial solution as an empty subgraph $T_0=(V,\emptyset)$ and choose an arbitrary vertex of $V$ to be the root $r$. We further initially label all of the vertices (including $r$) as \textit{active} vertices.  

We then perform $3/\epsilon$ rounds where for each round $i$, we let $T_i$ be our current partial solution and sample each vertex (besides the root $r$) from the remaining active vertices independently with probability $n^{-\epsilon}$ each, and \textit{merge} each unsampled vertex $u$ to the sampled vertex $v^*$ such that $d^{(h)}\left(u,v^*\right)$ is minimized. 
We merge an unsampled vertex $u$ to a sampled vertex $v^*$ by 
adding a minimum-weight $h$-length path $P_h(u,v^*)$ between $u$ and $v^*$ to our partial solution $T_i$, and marking $u$ as \textit{inactive} for the rest of the procedure. 
To avoid the event that no vertex is sampled we sample $r$ with probability $1$ each round, i.e.\ $r$ is always an active vertex. 
This process is repeated with the remaining active vertices in the following round. 


We will that show after $3/\epsilon$ rounds all the vertices besides the root will be inactive with high probability, implying that all vertices are connected the root and thus give a connected subgraph. We note that our ``with high probability'' guarantee is maintained if our algorithm ran for $c/\eps$ rounds for any constant $c\geq3$.
After the final round we simply return a shortest paths tree of $T_{3/\eps}$ that is 
rooted at $r$. Observe that this last step only increases the diameter of our subgraph by a constant factor $2$.  

The pseudocode for our algorithm is below.
\begin{algorithm}[H]
    \caption{$O(n^\epsilon/\epsilon)$-approximation with $O(1/\epsilon)$ length slack for length-constrained MST}
    \begin{algorithmic}[1]
        \label{alg:stars}
        \Statex \textbf{Input}: undirected graph $G=(V,E)$ with edge-weight function $w:E\to \mathbb{R}_{\geq0}$, diameter bound $h$, parameter $\epsilon\geq1/\poly(n)$ 
        \Statex \textbf{Output}: a subgraph of $G$ that is a spanning tree with diameter at most $O(1/\epsilon)\cdot h$ w.h.p.\ and has weight at most $O(n^\epsilon/\epsilon)\cdot\operatorname{OPT}_h$ in expectation
        \Statex
        \State $T_0\gets(V,\emptyset)$, $S_0\gets V$, choose an arbitrary $r\in V$ \Comment{Initialize partial solution and active vertices}
        \For{$i\in[3/\epsilon]$}
            \State $T_i\gets T_{i-1}$, $S_i\gets\{r\}$ 
            \For{$v\in S_{i-1}\setminus\{r\}$} \Comment{Sample vertices}
                \State $S_i\gets S_i\cup v$ with probability $n^{-\epsilon}$ 
            \EndFor
            \For{$u\in S_{i-1}\setminus S_i$} \Comment{Merge unsampled vertices}
                \State $v^*\gets\argmin_{v\in S_i}\left(d^{(h)}(u,v)\right)$
                \State $T_i\gets T_{i}\cup P_h(u,v^*)$ 
            \EndFor
        \EndFor
        \State Return a shortest paths tree of $T_{3/\epsilon}$ rooted at $r$
    \end{algorithmic}
\end{algorithm}

\section{Algorithm Analysis}\label{sec:analy}
In this section, we analyze \cref{alg:stars}, showing our main result. Specifically, we show the following. 
\begin{theorem}\label{thm:expectation}
    In polynomial time, \cref{alg:stars} finds a spanning tree with diameter at most $O(1/\epsilon)\cdot h$ with high probability and with weight at most $O(n^\epsilon/\epsilon)\cdot \operatorname{OPT}_h$ in expectation.
\end{theorem}
\noindent
By standard arguments, the expectation guarantee for the weight approximation of \cref{thm:expectation} can be made into a ``with high probability'' guarantee using an additional multiplicative $O(\log n)$ factor in runtime, proving \cref{thm:main}. The remainder of this section will be dedicated to proving \cref{thm:expectation}. We first show that \cref{alg:stars} gives a feasible solution with $O(1/\eps)$ length slack.
\begin{lemma}[Diameter]\label{lemma:randstarsfeasiblediameter}
    \cref{alg:stars} terminates with a spanning tree of diameter at most $O(1/\eps)\cdot h$ with high probability.
\end{lemma}
\begin{proof}
    We first show that the output is feasible, i.e.\ a spanning tree. It suffices to show that when the outer for loop terminates, all vertices but the root are inactive with high probability. Indeed, when all non-root vertices are inactive, that means we have some path connecting every inactive vertex to the root. 
    
    Thus, consider any non-root vertex $v$: the probability that it is active after $3/\epsilon$ rounds is $n^{\epsilon(-3/\epsilon)}=n^{-3}$. Union bounding over all $n-1$ non-roots gives that some non-root is active at the end with probability at most $n^{1-3}\leq n^{-2}$. Then with probability at least $1-n^{-2}$, we have that $T_{3/\eps}$ is a connected subgraph and thus any shortest paths tree on it is a spanning tree. 

    For the diameter of the solution, observe that any pair of unsampled vertices that are merged to the same sampled vertex must be connected via a path of length at most $2 \cdot i\cdot  h$ in $T_i$.
    This implies that when all non-roots are deactivated by round $3/\epsilon$ with high probability, the distance between any pair of vertices in $T_i$ is at most $2\cdot i\cdot h$. 
    Hence $T_{3/\epsilon}$ is a subgraph with diameter at most $O(1/\epsilon)\cdot h$ with high probability. Finally, taking a shortest paths tree of $T_{3/\eps}$ can only increase the diameter by another factor of $2$ since for any pair of vertices $u,v$ we can find a $O(1/\epsilon)\cdot h$-length path from $u$ to the root and then a $O(1/\epsilon)\cdot h$-length path from the root to $v$. 
\end{proof}
\noindent
The bulk of the algorithm analysis lies in the weight approximation, which we discuss next.

\subsection{Weight Analysis} 
Here we will prove that the expected weight of our final solution of \cref{alg:stars} is at most $O(n^\epsilon/\epsilon)\cdot \operatorname{OPT}_h$. The main step in proving this is to show that in each round $i$ of the for loop (which we will call a round) of \cref{alg:stars}, we add at most $O(n^\epsilon)\cdot \operatorname{OPT}_h$ to the total weight of our solution in expectation. In other words, we show that $\mathbb{E}[w(T_i)-w(T_{i-1})]\leq O(n^\eps)\cdot\operatorname{OPT}_h$.
 
We first give an overview of the weight analysis. The main idea here is defining a new graph based on Eulerian tours containing only edges of an optimal length-constrained MST solution of our instance (with duplicates) and a process of charging  to the weights of the edges of this new graph for each round of \cref{alg:stars}. In particular, we will show that in each round, the weight added by \cref{alg:stars} to the current solution is upper bounded by the weight added by the charging process (\cref{claim:med}), which is further upper bounded by $O(n^\epsilon)\cdot \operatorname{OPT}_h$ (\cref{claim:hard}).
We formalize the above intuition with several preliminaries before proceeding to the final lemma. We start by defining the graphs based on Eulerian tours and the charging process over these graphs that we will analyze. 

\begin{definition}[Eulerian Tour Cycle]\label{defn:etc}
    Given a tree $T$ and an Eulerian tour $\pi$ of $T$ of length $t:=2(n-1)$, let the Eulerian tour cycle of $\pi$ be a directed cycle $C=(V^\pi,E^\pi)$ with the following: \begin{enumerate}
        \item \textbf{Vertices:} $V^\pi=\{v_0, v_1, \ldots,v_{t-1}\}$ is the set of vertices of the tour $\pi$ in the order in which they are visited (where vertices of $T$ may be repeated). For every vertex $v$ of $T$ we use $v^\pi$ to denote the copy of $v$ in $V^\pi$ that has minimum index.
        \item \textbf{Edges:} For the pair of vertices $v_j,v_{j+1\pmod t}\in V^\pi$ corresponding to vertices $u,v\in T$, we define a directed edge $(j,j+1\pmod t)$ in $E^\pi$ with weight equal to the edge that connects $u,v$ in $\pi$, for every $j=0,1,\dots,t-1$.
    \end{enumerate} 
\end{definition}

\noindent
Observe that each vertex of $T$ appears at least once in $V^\pi$ and each edge of $T$ appears exactly twice in $E^\pi$. We will consider an Eulerian tour $\pi$ of an \textit{optimal} length-constrained MST solution of our given instance and its Eulerian tour cycle $C$. 

For sake of analysis it will be convenient to consider the Eulerian tour cycle on only active vertices without duplicates for a given round, motivating the following definition.

\begin{definition}[Contracted Eulerian Tour Cycle]\label{defn:cetc}
    In round $i$ we construct a contracted version of $C$ by defining the directed cycle $C_i=(V_i,E_i)$ with the following: 
    \begin{enumerate}
        \item \textbf{Contracted vertices}: $V_i=\{v^\pi:v\in S_{i-1}\cap V^\pi\}$ is the set of minimum-index vertex copies that are in $S_{i-1}$. We will further let $t_i:=|V_{i-1}|$ and reindex the vertices of $V_i$ as $v^\pi_0,v^\pi_1,\dots,v^\pi_{t_i-1}$ in the same relative order as their original indices in $V^\pi$.
        \item \textbf{Contracted edges}: we have a directed edge $(v^\pi_j,v^\pi_{j+1})\in E_i$ with weight equal to the sum of edge weights in the directed path from $v^\pi_j$ to $v^\pi_{j+1}$ in $C$ (note that these vertices may have different indices in $C$ than in $C_i$) for every $j=0,1,\dots,t_i-1$. 
    \end{enumerate}
\end{definition}
\noindent
Having defined these graphs, we are now ready to define a ``charging'' process whose weight in one round upper bounds the weight of \cref{alg:stars} in one round. Informally, in any round $i$ of \cref{alg:stars} where we add a set of paths between unsampled and sampled vertices to $T_i$ and incur the weight of those paths in $G$, we will map unsampled vertices to the sampled vertices in $C_i$ and charge the weights of the corresponding directed paths in $C_i$ (rather than in $G$). 

We may think of this as a merging process akin to that of \cref{alg:stars} since we also remove inactive vertices from consideration for the $C_{i+1}$ charging. However, in this new process we merge on $C_i$ and do not necessarily consider minimum-weight $h$-length paths in $G$. We overload notation and let  $w_i(v^\pi_j,v^\pi_k)$ be the weight of the directed path from $v^\pi_j$ to $v^\pi_k$ in $C_i$. The charging process is defined below:

\begin{definition}[Eulerian Tour Charging]\label{defn:etcharging}
    For every round $i$ of \cref{alg:stars}, we define the following charging procedure on $C_i$: for every $v^\pi_j\in V_{i-1}\setminus V_i$ (i.e.\ for every unsampled vertex) let $\nextme_i(v)$ be the vertex $v^\pi_k$ in $V_i$ (where $k$ is its index in $V_{i-1}$) that minimizes $k-j\pmod{t_i}$ and define the function
    $$ \Phi(C_i):=\sum_{v\in V_{i-1}\setminus V_i}w_i(v,\nextme_i(v)) $$ 
\end{definition}

\noindent
By the process described in \cref{defn:etcharging}, each unsampled vertex is essentially mapped to a directed path in $C$ (and hence in $C_i$ as well). 
Since we guarantee that the root is always a sampled vertex, every unsampled vertex gets mapped to some sampled vertex in $C_i$ because it is a cycle. It is also clear that the directed paths whose weights we are summing 
to $\Phi(C_i)$ have length at most $h$, since a path in $C_i$ corresponds to a simple path in the optimal length-constrained MST solution which by definition has diameter at most $h$. We will typically refer to an Eulerian tour charging process as $C$ charging, or $C_i$ charging if we are specifically referring to round $i$. 


See \Cref{fig:eulertourcyclefull} for an example of the previous definitions. 

\begin{figure}[H]
    \begin{subfigure}[b]{.33333\textwidth}
      \centering
      \includegraphics[width=.65\linewidth]{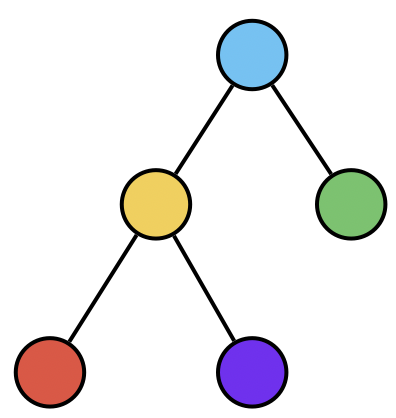}
      \caption{A tree $T$.}\label{subfig:tree}
    \end{subfigure}%
    ~
    \begin{subfigure}[b]{.33333\textwidth}
      \centering
      \includegraphics[width=.9\linewidth]{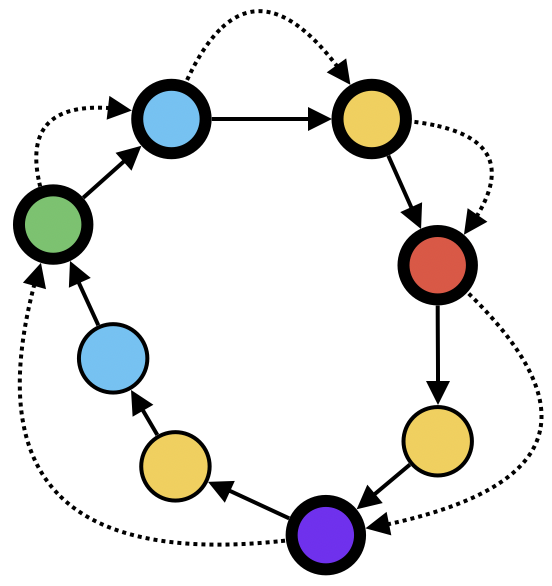}
      \caption{$C$ and $C_0$.}\label{subfig:et1}
    \end{subfigure}%
    ~
    \begin{subfigure}[b]{.33333\textwidth}
        \centering
        \includegraphics[width=1.025\linewidth]{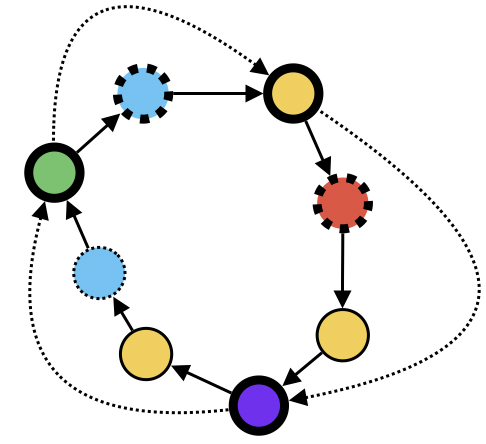}
        \caption{$C_1$.}\label{subfig:et2}
    \end{subfigure}
\caption{Examples of \cref{defn:etc,defn:cetc,defn:etcharging}. \Cref{subfig:tree}: a tree $T$; further let $\pi$ be an Eulerian tour of $T$ in the following order: blue, yellow, red, yellow, purple, yellow, blue, and green. \Cref{subfig:et1}: an Eulerian tour cycle $C$ (in solid edges) and its contracted Eulerian tour cycle (in dotted edges). The bolded vertices are minimum index copies. \Cref{subfig:et2}: the contracted Eulerian tour cycle of $\pi$ in the first round $C_1$ when $S_1$ contains only green, yellow, and purple. In the $C_1$ charging, the bolded blue maps to the bolded yellow, the bolded red maps to the bolded purple, and $\Phi(C_i)$ is the sum of the weights of the corresponding paths in dotted edges.}
\label{fig:eulertourcyclefull}
\end{figure}

\noindent
We will use the following two claims regarding the Eulerian tour charging process. In the first, we show that the weight added by \cref{alg:stars} to the current solution in one round is at most the Eulerian tour's charging weights in one round. 

\begin{claim}\label{claim:med}
    For any round $i$ of \cref{alg:stars}, we have $w(T_{i})-w(T_{i-1})\leq \Phi(C_i)$.
\end{claim}
\begin{proof}
    In \cref{alg:stars}, we always merge each unsampled vertex $u$ to the sampled vertex that is connected to $u$ with the cheapest minimum-weight $h$-length path, and these are the paths we add to our solution. In the $C$ charging, we merge unsampled vertices to sampled vertices based on the order of their appearances in $V^\pi$. In other words, the paths we charge in $C_i$ are paths that exist in $G$, but are not necessarily minimum-weight $h$-length paths in $G$. Then for any unsampled vertex $v^\pi_j\in V_{i-1}\setminus V_i$ we have $d^{(h)}(v_j,v^*)\leq w_i(v^\pi_j,\nextme_i(v^\pi_j))$. Hence the sum of the weights of the paths that we add in \cref{alg:stars} is at most the sum of the weights of the directed paths that we add in the $C$ charging for any round. 
\end{proof}
\noindent
For the next claim, the following fact about sums of exponentials will be useful. 
    \begin{restatable*}{fact}{factsumexps}\label{fact:sumofexps}
        Given $M,N>0$ we have
    $\sum_{j=0}^{N-1} \sum_{k=0}^{N-1} e^{(-j-k)/M} \leq O(M^2)$.
    \end{restatable*}
\noindent
We defer the proof of \cref{fact:sumofexps} to \cref{section:appendix}. We now compute an upper bound of $\mathbb{E}\left[\Phi(C_i)\right]$. 
\begin{claim}\label{claim:hard} 
    For any round $i$, we have $\mathbb{E}\left[\Phi(C_i)\right]\leq O\left(n^\eps\right) \cdot \operatorname{OPT}_h$.
    \end{claim}
\begin{proof}
    We remind the reader that $S_i$ denotes the set of sampled vertices in round $i$ out of the remaining active vertices $S_{i-1}$. For further brevity and abuse of notation we will use $w_i(v^\pi_j,S_i)$ to denote the weight of the path from $v^\pi_j$ to $\operatorname{next}(v^\pi_j)$ in $C_i$. 
        
    To bound the expected weight of the $C_i$ merging, we first write it as 
    \begin{equation}\label{equ:mainsum}
        \mathbb{E}\left[\Phi(C_i)\right] = \mathbb{E}\left[\sum_{j=0}^{t_i-1}w_i\left(v^\pi_j,S_i\right)\right] = \sum_{j=0}^{t_i-1}\mathbb{E}\left[w_i\left(v^\pi_j,S_i\right)\right]
    \end{equation}
    where the second equality is by linearity of expectation. From here on all addition in indices of a vertex will be done over $\mathbb{Z}_{t_i}$. Fixing any $j$, we can bound each term in the sum of \cref{equ:mainsum} as follows:
    \begin{align*}
        \mathbb{E}\left[w_i\left(v^\pi_j,S_i\right)\right] &\leq \sum_{k=0}^{t_i-1} \operatorname{Pr}\left(v^\pi_j,v^\pi_{j+1},\dots,v^\pi_{j+k-1} \not\in S_i \wedge v^\pi_{j+k} \in S_i\right)\cdot w_i\left(v^\pi_j,v^\pi_{j+k}\right) \\
        &= \sum_{k=0}^{t_i-1} \operatorname{Pr}\left(v^\pi_j,v^\pi_{j+1},\dots,v^\pi_{j+k-1} \not\in S_i \wedge v^\pi_{j+k} \in S_i\right)\cdot \sum_{k'=0}^{k-1} w_i\left(v^\pi_{j+k'},v^\pi_{j+k'+1}\right)\\
        &= \sum_{k=0}^{t_i-1} \frac{1}{n^\epsilon}\left(1-\frac{1}{n^\epsilon}\right)^{k}\cdot \sum_{k'=0}^{k-1} w_i\left(v^\pi_{j+k'},v^\pi_{j+k'+1}\right) \\
        &\leq \frac{1}{n^\epsilon}\sum_{k=0}^{t_i-1}e^{-k/n^\epsilon}\cdot\sum_{k'=0}^{k-1} w_i\left(v^\pi_{j+k'},v^\pi_{j+k'+1}\right)
    \end{align*}
    where the second line is expanding the weight of the path into a sum of the weights of the path's edges, the third line is by our sampling process, and the fourth line is by the inequality $1-x\leq e^{-x}$. 
        
    Then we can put this back into \cref{equ:mainsum} and switch the order of summation to get 
    \begin{align*}
        \sum_{j=0}^{t_i-1}\mathbb{E}\left[w_i\left(v^\pi_j,S_i\right)\right] 
        &\leq \frac{1}{n^\epsilon}\sum_{k'=0}^{t_i-1} w_i\left(v^\pi_{k'},v^\pi_{k'+1}\right)\cdot \sum_{j=0}^{t_i-1} \sum_{k=0}^{t_i-1}  e^{(-j-k)/n^\epsilon} 
    \end{align*}
    Observe this rearrangement is counting the weight of each edge for each possible $j,k$ path. Now it remains to bound $\sum_{j=0}^{t_i-1} \sum_{k=0}^{t_i-1}  e^{(-j-k)/n^\epsilon}$. 
    Intuitively, every edge
    $(v^\pi_{k'},v^\pi_{k'+1})$ contributes its weight to $\Phi(C_i)$ in round $i$ at most $\sum_{j=0}^{t_i-1}\sum_{k=0}^{t_i-1} e^{(-j-k)/n^\epsilon}$ times (in expectation) since we have to consider this edge's contribution to every $j,j+k$ path for every $0\leq j<t_i$ and $0\leq k<t_i$. Specifically, applying \cref{fact:sumofexps} with $N=t_i,M=n^\epsilon$ gives the following:
    \begin{align*}
        \sum_{j=0}^{t_i-1}\mathbb{E}\left[w_i(v^\pi_j,S_i)\right] &\leq  \frac{1}{n^\epsilon}\sum_{k'=0}^{t_i-1} w_i\left(v^\pi_{k'},v^\pi_{k'+1}\right)\cdot \sum_{j=0}^{t_i-1} \sum_{k=0}^{t_i-1}  e^{(-j-k)/n^\epsilon} \\
        &\leq \frac{1}{n^\epsilon}\sum_{k'=0}^{t_i-1}w_i\left(v^\pi_{k'},v^\pi_{k'+1}\right)\cdot O\left(n^{2\epsilon}\right) \\
        &\leq O\left(n^\epsilon\right)\sum_{k'=0}^{t_i-1}w_i\left(v^\pi_{k'},v^\pi_{k'+1}\right) \\
        &\leq O\left(n^\epsilon\right)\cdot \operatorname{OPT}_h
    \end{align*}
    where the fourth line is because $\sum_{k'=0}^{t_i-1}w_i(v^\pi_{k'},v^\pi_{k'+1})\leq \sum_{e\in E^\pi}w(e) =  2\cdot\operatorname{OPT}_h$. 
    
    Thus, each edge contributes at most $O(n^\epsilon)$ times its own weight to the total weight of the $C_i$ charging in expectation, and we have that the weight of the $C_i$ charging for any round $i$ is at most $O(n^\epsilon)\cdot\operatorname{OPT}_h$ in expectation.
\end{proof}
\noindent
With this, we have what we need to bound the expected weight of \cref{alg:stars}.

\begin{lemma}[Weight]\label{lemma:randstarscost}
    \cref{alg:stars} terminates with a spanning tree with weight at most $O(n^\epsilon/\epsilon)\cdot \operatorname{OPT}_h$ in expectation.
\end{lemma}
\begin{proof}
    By \cref{claim:med,claim:hard}, we have that in any round $i$ the expected increase in weight of the partial solution of \cref{alg:stars} is at most $O(n^\epsilon)\cdot\operatorname{OPT}_h$. The number of such increases is $3/\epsilon$ (i.e.\ the number of rounds), and taking a shortest paths tree on $T_{3/\epsilon}$ cannot increase the weight. Therefore the weight of our final solution is at most $O(n^\epsilon/\epsilon)\cdot \operatorname{OPT}_h$ in expectation.
\end{proof}

\subsection{Putting Things Together}
Observe that each step of \cref{alg:stars} can be done in polynomial time, and there are polynomially many rounds by our choice of $\epsilon$. Hence \cref{alg:stars} is a polynomial-time algorithm. Combining this with \cref{lemma:randstarsfeasiblediameter,lemma:randstarscost} yields \cref{thm:expectation}, and as stated in the beginning of this section, \cref{thm:main} follows from there.

\section*{Acknowledgments}
The first author would like to thank R Ravi for introducing him to this problem and for many useful conversations about it over the years.

\bibliographystyle{alpha}
\bibliography{references}
\appendix
\section{Proof of \cref{fact:sumofexps}}\label{section:appendix}
\factsumexps
\begin{proof}
    For $a\geq0$, let $B_a$ be the interval $[a\cdot M, (a+1)\cdot M)$. We use these $M$-length ``buckets'' to more conveniently bound the contribution by the $M$ values that lie within each interval. 
    We can then bound the double sum as follows:
    \begin{align*}
        \sum_{j=0}^{N-1}\sum_{k=0}^{N-1}e^{(-j-k)/M} &= \sum_{a=0}^{N/M-1} \sum_{j\in B_a } \sum_{b=0}^{N/M-1} \sum_{k\in B_b} e^{(-j-k)/M} \\
        &\leq \sum_{a=0}^{N/M-1} \sum_{b=0}^{N/M-1} |B_a|\cdot|B_b|\cdot e^{-a-b} \\
        &= M^2\sum_{a=0}^{N/M-1} e^{-a} \sum_{b=0}^{N/M-1} e^{-b} \\
        &= O(M^{2})
    \end{align*}
    where the first line is breaking the interval $[0,N)$ into buckets $B_a,B_b$ for every $a,b$, the second line is because $(j+k)/M\geq a+b\implies e^{(-j-k)/M}\leq e^{-a-b}$ for every $j\in B_a,k\in B_b$, and the last line is because both of the geometric sums are bounded by constants.
\end{proof}

\end{document}